\documentclass[11pt]{article}

\usepackage{mathptmx}
\usepackage{pstricks,pst-text,pst-node}
\usepackage{epsfig,color}
\usepackage{amsmath,amsthm,amssymb}
\usepackage{algorithm,algorithmic}

\usepackage[letterpaper,margin=2.5cm,centering]{geometry}

\newtheorem{theorem}{Theorem}
\newtheorem{lemma}[theorem]{Lemma}
\newtheorem{claim}[theorem]{Claim}
\newtheorem{corollary}[theorem]{Corollary}
\newtheorem{definition}[theorem]{Definition}

\newtheorem{fact}[theorem]{Fact}

\newcommand{\ignore}[1]{}

\newcommand{\ess}{\ensuremath{\mathcal{S}}}
\newcommand{\cee}{\ensuremath{\mathcal{C}}}
\newcommand{\fC}{\ensuremath{\mathcal{K}}}
\newcommand{\calR}{\ensuremath{\mathcal{R}}}
\newcommand{\R}{\ensuremath{\mathcal{R}}}

\newcommand{\ST}{\ensuremath{\mathtt{ST}_{G,R}}}
\newcommand{\opt}{\ensuremath{\mathtt{opt}}}

\newcommand{\mst}{\ensuremath{\mathtt{mst}}}

\newcommand{\smst}{\ensuremath{\overline{\mathtt{mst}}}}
\newcommand{\MST}{\ensuremath{\mathtt{MST}}}
\newcommand{\NP}{\ensuremath{\text{NP}}}

\newcommand{\IR}{\mathbb{R}}

\newcommand{\rc}{\ensuremath{\mathtt{rc}}}
\newcommand{\rz}{\ensuremath{\mathtt{RZ}}}

\newcommand{\sppart}[1][]{\ensuremath{\Pi^{#1}}}
\newcommand{\LOSS}{\ensuremath{\mathtt{L}}}
\newcommand{\loss}{\ensuremath{\mathtt{l}}}
\newcommand{\rs}{\ensuremath{\bar{r}}}
\newcommand{\bs}{\ensuremath{\backslash}}
\newcommand{\newLP}{\ensuremath{\mathrm{(P}_{ST}^\emptyset\mathrm{)}}}
\newcommand{\newLPsp}{\ensuremath{\mathrm{(P}_{ST}^\emptyset\mathrm{)\,}}}

\newcommand{\foob}{\ensuremath{\mathrm{(P}_{ST}^{\ess'}\mathrm{)}}}

\newcommand{\doob}{\ensuremath{\mathrm{(D}_{ST}^{\ess^*}\mathrm{)}}}
\newcommand{\doobsp}{\ensuremath{\mathrm{(D}_{ST}^{\ess^*}\mathrm{)\,}}}
\newcommand{\poob}{\ensuremath{\mathrm{(P}_{ST}^{\ess^*}\mathrm{)}}}
\newcommand{\poobsp}{\ensuremath{(\mathrm{P}_{ST}^{\ess^*}\mathrm{)\,}}}
\newcommand{\REF}[1]{\mathrm{(}#1\mathrm{)}}
\newcommand{\K}{\ensuremath{\fC_r}}

\newcommand{\IS}[1]{#1} 
\newcommand{\ISC}[1]{#1} 

\newcommand{\steiner}[2]{
  \fnode[framesize=0.3](#1){#2}
}

\newcommand{\terminal}[2]{
  \cnode[fillstyle=solid,fillcolor=black](#1){.18}{#2}
}

\newrgbcolor{bg}{.85 .85 .85}
\newrgbcolor{bg-dark}{.75 .75 .75}

\newcounter{ntheorem}
\newcounter{nntheorem}
\newcounter{gtheorem}
\newcounter{ggtheorem}

\begin{document}

\title{A Partition-Based Relaxation For Steiner Trees}

\author{
  Jochen K\"onemann\thanks{{
    Department of Combinatorics and Optimization, University of Waterloo,
    200 University Avenue West, Waterloo, ON N2L 3G1, Canada.
    Email: \{jochen,dagpritc,ktan\}@math.uwaterloo.ca}
}
  \and
  David Pritchard$^*$
  \and
  Kunlun Tan$^*$
}

\date{\today}

\maketitle

\begin{abstract}
 The Steiner tree problem is a classical \NP-hard
optimization problem with a wide range of practical applications. In
an instance of this problem, we are given an undirected graph
$G=(V,E)$, a set of {\em terminals} $R\subseteq V$, and non-negative
costs $c_e$ for all edges $e \in E$. Any tree that contains all
terminals is called a \emph{Steiner tree}; the goal is to find a
minimum-cost Steiner tree. The nodes $V \bs R$ are called
\emph{Steiner nodes}.

The best approximation algorithm known for the Steiner tree problem
is due to Robins and Zelikovsky (SIAM J. Discrete Math, 2005); their
\emph{greedy} algorithm achieves a performance guarantee of
$1+\frac{\ln 3}{2} \approx 1.55$.  The best known \emph{linear
programming} (LP)-based algorithm, on the other hand, is due to Goemans
and Bertsimas (Math. Programming, 1993) and achieves an
approximation ratio of $2-2/|R|$. In this paper we establish a link
between greedy and LP-based approaches by showing that Robins and
Zelikovsky's algorithm has a natural primal-dual interpretation with
respect to a novel \emph{partition}-based linear programming
relaxation.  We also exhibit surprising connections between the new
formulation and existing LPs and we show that the new LP is stronger
than the bidirected cut formulation.

An instance is \emph{$b$-quasi-bipartite} if each connected component
of $G \bs R$ has at most $b$ vertices. We show that Robins' and
Zelikovsky's algorithm has an approximation ratio better than
$1+\frac{\ln 3}{2}$ for such instances, and we prove that the
integrality gap of our LP is between $\frac{8}{7}$ and
$\frac{2b+1}{b+1}$.
\end{abstract}

\section{Introduction}\label{sec:intro}

\ignore{In an instance of the {\em Steiner tree} problem, we are given an
undirected graph $G=(V,E)$, a set $R \subseteq V$ of {\em terminals}
and a non-negative cost $c_e$ for each edge $e \in E$. The
vertices in the set $V\setminus R$ are referred to as {\em Steiner
vertices}. The goal is to find a minimum-cost tree $T$ that contains
all terminals.}

The Steiner tree problem is a classical problem in combinatorial
optimization which owes its practical importance to a host of
applications in areas as diverse as VLSI design and computational
biology. The problem is \NP-hard~\cite{Ka72}, and Chleb{\'i}k and
Chleb{\'i}kov{\'a} show in \cite{CC02} that it is \NP-hard even to
\emph{approximate} the minimum-cost Steiner tree within any ratio
better than $\frac{96}{95}$.  They also show that it is \NP-hard to
obtain an approximation ratio better than $\frac{128}{127}$ in {\em
quasi-bipartite} instances of the Steiner tree problem. These are
instances in which no two Steiner vertices are adjacent in the
underlying graph $G$.

\subsection{Greedy algorithms and $r$-Steiner trees}\label{sec:rSteiner}

\ignore{
Probably the first known approximation algorithm for the Steiner
tree problem is the {\em minimum-spanning tree heuristic} which is
generally attributed to Moore. Moore's algorithm appears in the
literature as early as 1968~\cite{GP68}, and it is known to have a
performance guarantee of $2$. This remained practically the best
known result for more than 20 years until, in 1990, Zelikovsky
suggested computing Steiner trees with a special structure, so
called {\em $r$-Steiner trees}.}

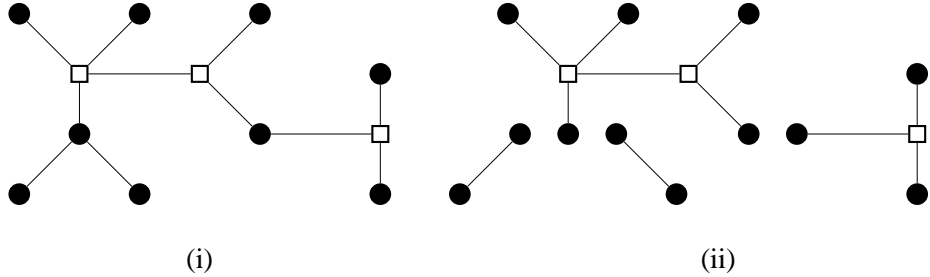
\begin{figure}
\begin{center} \leavevmode
\begin{pspicture}(0.8,0)(12.8,3.2)

\psset{unit=.8}

\terminal{1,4}{t1} \terminal{1,1}{t2} \terminal{3,1}{t3}
\terminal{3,4}{t4} \terminal{5,4}{t5} \terminal{5,2}{t6}
\terminal{7,3}{t7} \terminal{7,1}{t8} \terminal{2,2}{t9}

\steiner{2,3}{s1}  \steiner{4,3}{s3} \steiner{7,2}{s4}

\ncline[linewidth=.1mm]{t1}{s1} \ncline[linewidth=.1mm]{t2}{t9}
\ncline[linewidth=.1mm]{t3}{t9} \ncline[linewidth=.1mm]{t4}{s1}
\ncline[linewidth=.1mm]{t5}{s3} \ncline[linewidth=.1mm]{t6}{s3}
\ncline[linewidth=.1mm]{t7}{s4} \ncline[linewidth=.1mm]{t8}{s4}
\ncline[linewidth=.1mm]{t6}{s4} \ncline[linewidth=.1mm]{s1}{t9}
\ncline[linewidth=.1mm]{s1}{s3}

\uput{.5}[270](4,0.6){(i)}


\terminal{9,4}{at1} \terminal{8.2,1}{at2} \terminal{11.8,1}{at3}
\terminal{11,4}{at4} \terminal{13,4}{at5} \terminal{13,2}{at6}
\terminal{13.8,2}{bt6} \terminal{15.8,3}{bt7} \terminal{15.8,1}{bt8}
\terminal{9.2,2}{at9} \terminal{10.8,2}{bt9} \terminal{10,2}{ct9}

\steiner{10,3}{as1}  \steiner{12,3}{as3} \steiner{15.8,2}{bs4}

\ncline[linewidth=.1mm]{at1}{as1} \ncline[linewidth=.1mm]{at3}{bt9}
\ncline[linewidth=.1mm]{at2}{at9} \ncline[linewidth=.1mm]{at4}{as1}
\ncline[linewidth=.1mm]{at5}{as3} \ncline[linewidth=.1mm]{at6}{as3}

\ncline[linewidth=.1mm]{bt7}{bs4} \ncline[linewidth=.1mm]{bt8}{bs4}
\ncline[linewidth=.1mm]{bt6}{bs4} \ncline[linewidth=.1mm]{as1}{ct9}
\ncline[linewidth=.1mm]{as1}{as3}

\uput{.5}[270](12.5,0.6){(ii)}

\end{pspicture}
\caption{\label{fig:decomp} The figure shows a Steiner tree in (i)
and its decomposition into full components in (ii). Square and round
nodes correspond to Steiner and terminal vertices, respectively.
This particular tree is 5-restricted.}
\end{center}
\end{figure}

One of the first approximation algorithms for the Steiner tree
problem is the well-known {\em minimum-spanning tree heuristic}
which is widely attributed to Moore~\cite{GP68}. Moore's algorithm
has a performance ratio of $2$ for the Steiner tree problem and this
remained the best known until the 1990s, when Zelikovsky~\cite{Ze93}
suggested computing Steiner trees with a special structure, so
called {\em $r$-Steiner trees}. Nearly all of the Steiner tree
algorithms developed since then use $r$-Steiner trees. We now
provide a formal definition.

A {\em full Steiner component} (or {\em full component} for short)
is a tree whose internal vertices are Steiner vertices, and whose
leaves are terminals. The edge set of any Steiner tree can be
partitioned into full components, by {\it splitting} the tree at
terminals: see Figure \ref{fig:decomp} for an example. An
\emph{$r$-(restricted)-Steiner tree} is defined to be a Steiner tree
all of whose full components have at most $r$ terminals. We remark
that such a Steiner tree may in general not exist; for example, if
$G$ is a star with a Steiner vertex at its center and more than $r$
terminals at its tips. To avoid this problem, each Steiner vertex
$v$ is \emph{cloned} sufficiently many times: introduce copies of
$v$ and connect these copies to all of $v$'s neighbors in the graph.
Copies of an edge have the same cost as the corresponding original
edge in $G$.

Let $\opt$ and $\opt_r$ be the cost of an optimum Steiner tree and
that of an optimal $r$-Steiner tree, respectively, for the given
instance.  Define the {\em $r$-Steiner ratio} $\rho_r$ as the supremum
of $\opt_r/\opt$ over all instances of the Steiner tree problem.  In
\cite{BD97}, Borchers and Du provided an exact characterization of
$\rho_r$. The authors showed that $\rho_r = 1 + \Theta(1/\log r)$ and
hence that $\rho_r$ tends to 1 as $r$ goes to infinity.

Computing minimum-cost $r$-Steiner trees is \NP-hard for $r \geq
4$~\cite{GJ77}, even if the underlying graph is quasi-bipartite.  The
complexity status for $r=3$ is unresolved, and the case $r=2$ reduces
to the minimum-cost spanning tree problem.

In \cite{Ze93}, Zelikovsky used $3$-restricted full components to
obtain an $11/6$-approximation for the Steiner tree problem.
Subsequently, a series of papers (e.g., \cite{BR94,HP99,KZ97,PS00})
improved upon this result. These efforts culminated in a recent
paper by Robins and Zelikovsky~\cite{RZ05} in which the authors
presented a $\left(1+\frac{\ln 3}{2}\right) \approx
1.55$-approximation (subsequently referred to as \rz) for the
$r$-Steiner tree problem. They hence obtain, for each fixed $r \geq
2$, a $1.55\rho_r$ approximation algorithm for the (unrestricted)
Steiner tree problem. We refer the reader to two surveys in
\cite{GH+01b,PS02}.

\subsection{Approaches based on linear programs}

There is a large body of work on linear programming (LP)-based
approximation algorithms for problems in combinatorial optimization.
First, one finds a {\it good} LP relaxation for the problem. Then
one designs an algorithm that produces a feasible integral solution
whose cost is provably close to that of an optimum fractional
solution for this relaxation. Many aspects of different LP
relaxations for the Steiner tree problem have been investigated
(e.g., \cite{An80,CR94a,CR94b,DB+01,E67,GM93,PVd01,War98,Wo84}).

Many of these LPs have been fruitfully used in \emph{integer
programming}-based approaches to exactly solve instances of up to
ten thousand nodes \cite{PVd01a}. Another common area in which LPs
are useful is the design of polynomial time approximation algorithms
via the \emph{primal-dual method} (e.g., \cite{GW97}). In this
method, a feasible solution of the relaxation's LP dual is used to
obtain a lower bound on the optimum cost.

The ``classical" LP-based approximation algorithms for Steiner
trees~\cite{GB93} and forests~\cite{AKR} use the \emph{undirected
cut relaxation} \cite{An80} and have a performance guarantee of $2 -
\frac{2}{|R|}$. This relaxation has an integrality gap of $2 -
\frac{2}{|R|}$ and the analysis of these algorithms is therefore
tight. Slightly improved algorithms have since been designed
 \cite{K+05, PVd00} but do not achieve any constant approximation
factor better than 2.

In the special case of quasi-bipartite graphs, Rajagopalan and
Vazirani~\cite{RV99} and Rizzi~\cite{Ri03} obtained a $\frac{3}{2}$
approximation for the Steiner tree problem in quasi-bipartite
graphs. The analysis of \cite{RV99} applies the primal-dual method
to the {\em bidirected cut relaxation}~\cite{E67,Wo84}. The
bidirected cut relaxation is widely conjectured to have a worst-case
integrality gap that is close to $1$: the worst known example shows
a gap of only $\frac{8}{7}$ (see Section \ref{sec:gap}). Despite its
conjectured strength, this new relaxation has not yet given rise to
a Steiner tree algorithm with performance guarantee better than $2$
in general graphs.

\subsection{Contribution of this paper}\label{sec:contribution}
In this paper we provide algorithmic evidence that the primal-dual
method is useful for the Steiner tree problem. We first present a
novel LP relaxation for the Steiner tree problem. It uses full
components to strengthen a formulation based on {\em Steiner
partition} inequalities~\cite{CR94a}. We then show that the
algorithm \rz\ of Robins and Zelikovsky can be analyzed as a
primal-dual algorithm using this relaxation. 
We can show (see Section \ref{sec:gap}) that our relaxation is
strictly stronger than the \emph{standard} Steiner partition
formulation; so the use of full components strengthens the partition
inequalities.

In \cite{RZ05}, Robins and Zelikovsky showed that \rz\ has a
performance ratio of $1.279$ for quasi-bipartite graphs, and a
performance ratio of $1.55$ in general graphs. We prove a natural
interpolation of these two results. For a Steiner vertex $v$, define
its {\em Steiner neighborhood} $S_v$ to be the collection of
vertices that are in the same connected component as $v$ in $G \bs
R$. \ignore{
 set of Steiner vertices that are
reachable from $v$ in $G$ by paths containing only Steiner vertices,
i.e.,
$$ S_v = \{ u \in V\setminus R \,:\, \mbox{ there is a } u,v\mbox{-path } P \mbox{ in } G
                  \mbox{ s.t. } V(P) \cap R = \emptyset \}. $$
} A graph is {\em $b$-quasi-bipartite} if all of its Steiner
neighborhoods have cardinality at most $b$. Note,
``1-quasi-bipartite'' is synonymous with ``quasi-bipartite.'' We prove:

\setcounter{nntheorem}{\value{theorem}}
\begin{theorem}\label{thm:main}
Given an undirected, $b$-quasi-bipartite graph $G=(V,E)$, terminals
$R \subseteq V$, and a fixed constant $r \geq 2$, Algorithm \rz\
returns a feasible Steiner tree $T$ s.t.
$$
  c(T) \leq \left\{\begin{array}{l@{\quad : \quad}l}
      1.279 \cdot \opt_r & b=1 \\
      (1+\frac{1}{e}) \cdot \opt_r & b\in\{2, 3, 4\} \\
      \left(1+\frac{1}{2}\ln\left(3-\frac2b\right)\right)\opt_r & b \geq 5.
      \end{array}\right.
$$
\end{theorem}
Unfortunately, Theorem
\ref{thm:main} does not imply that our new relaxation has a small
integrality gap.  Nonetheless, we obtain the following bounds, when
$G$ is $b$-quasi-bipartite: \setcounter{ggtheorem}{\value{theorem}}
\begin{theorem}\label{thm:gap}
Our new relaxation has an integrality gap between $\frac{8}{7}$ and
$\frac{2b+1}{b+1}$.
\end{theorem}

\ignore{ The first part of the next section focuses on background
material related to spanning trees. We then introduce our new
LP-relaxation for the Steiner tree problem. In Section
\ref{sec:new-relax} we show that the new relaxation is equivalent to
a slightly strengthened version of the bidirected cut relaxation.
Section \ref{sec:zel} describes a primal-dual view of Robins and
Zelikovsky's greedy Steiner tree algorithm. Its analysis is given in
Sections \ref{sec:ana} and \ref{sec:tpcost}. Finally, in Section
\ref{sec:gap} we present a proof of Theorem \ref{thm:gap}.}

\section{Spanning trees and a new LP relaxation for Steiner trees}
\label{sec:mst}

Our work is strongly motivated by, and uses, results on the spanning
tree polyhedron due to Chopra~\cite{C89}. In this section, we first
discuss Chopra's characterization of the spanning tree polyhedron;
then we mention a primal-dual interpretation of Kruskal's spanning
tree algorithm~\cite{Kr56} based on Chopra's formulation. Finally we
extend ideas in \cite{CR94a,CR94b} to derive a new LP relaxation for
the Steiner tree problem.

\subsection{The spanning tree polyhedron}

To formulate the minimum-cost spanning tree (MST) problem as an LP,
we associate a variable $x_e$ with every edge $e \in E$. Each
spanning tree $T$ corresponds to its \emph{incidence vector} $x^T$,
which is defined by $x_e^T = 1$ if $T$ contains $e$ and $x_e^T = 0$
otherwise. Let $\Pi$ denote the set of all partitions of the vertex
set $V$, and suppose that $\pi \in \Pi$. The \emph{rank} $r(\pi)$ of
$\pi$ is the number of parts of $\pi$. Let $E_\pi$ denote the set of
edges whose ends lie in different parts of $\pi$. Consider the
following LP.

\begin{align}
\min \quad & \sum_{e \in E} c_e x_e \tag{$\mathrm{P}_{SP}$} \label{p-sp} \\
\mbox{s.t.} \quad &  \sum_{e \in E_{\pi}} x_e \geq r(\pi)-1 \quad
\forall \pi \in {\Pi},  \notag\\
 & x \geq 0. \notag
\end{align}

Chopra \cite{C89} showed that the feasible region of \eqref{p-sp} is
the convex hull of all incidence vectors of spanning trees, and hence
each basic optimal solution corresponds to a minimum-cost spanning
tree. Its dual LP is

\begin{align}
\max \quad & \sum_{\pi \in {\Pi}} (r(\pi)-1)\cdot y_{\pi} \tag{$\textrm{D}_{SP}$} \label{d-sp} \\
\mbox{s.t.} \quad  & \sum_{\pi: e \in E_{\pi}} y_{\pi} \leq c_e
\quad
        \forall e \in E, \label{d-sp:edge}\\
& y \geq 0. \label{d-sp-nn}
\end{align}

\subsection{A primal-dual interpretation of Kruskal's MST algorithm}
\label{sec:kruskal} Kruskal's algorithm can be viewed as a
continuous process over {\em time}: we start with an empty tree at
time $0$ and add edges as time increases. The algorithm terminates
at time $\tau^*$ with a spanning tree of the input graph $G$.  In
this section we show that Kruskal's method can be interpreted as a
primal-dual algorithm (see also \cite{GW97}).  At any time $0 \leq
\tau \leq \tau^*$ we keep a pair $(x_{\tau},y_{\tau})$, where
$x_{\tau}$ is a partial (possibly infeasible) 0-1 primal solution
for \eqref{p-sp} and $y_{\tau}$ is a feasible dual solution for
\eqref{d-sp}. Initially, we let $x_{e,0}=0$ for all $e \in E$ and
$y_{\pi,0}=0$ for all $\pi \in \Pi$.

Let $G_{\tau}$ denote the forest corresponding to partial solution
$x_{\tau}$ and let $E_\tau$ denote its edges, i.e., $E_{\tau}=\{e
\in E \mid x_{e,\tau}=1\}$. We then denote by $\pi_{\tau}$ the
partition induced by the connected components of $G_{\tau}$. At time
$\tau$, the algorithm then increases $y_{\pi_{\tau}}$ until a
constraint of type \eqref{d-sp:edge} for edge $e \in E\setminus
E_{\pi_\tau}$ becomes tight. Assume that this happens at time
$\tau'> \tau$. The dual update is
$$
y_{\pi_{\tau},\tau'} = \tau'-\tau.
$$
We then include $e$ in our solution, i.e., we set $x_{e,\tau'}=1$.
If more than one edge becomes tight at time $\tau'$, we can process
these events in any arbitrary order. Thus, note that we can pick any
such tight edge first in our solution. We terminate when $G_\tau$ is
a spanning tree. Chopra~\cite{C89} showed that the final primal and
dual solutions have the same objective value (and are hence
optimal), and we give a proof of this fact for completeness.

\begin{theorem}\label{thm:mst}
At time $\tau^*$, algorithm \MST\ finishes with a pair
$(x_{\tau^*},y_{\tau^*})$ of primal and dual feasible solutions to
\eqref{p-sp} and \eqref{d-sp}, respectively, such that
$$
\sum_{e \in E} c_e x_{e,\tau^*} = \sum_{\pi \in \Pi} (r(\pi)-1)
\cdot y_{\pi,\tau^*}.
$$
\end{theorem}
\begin{proof}
Notice that for all edges $e \in E_{\tau^*}$ we must have
$c_e=\sum_{\pi: e \in E_{\pi}} y_{\pi,\tau^*}$ and hence, we can
express the cost of the final tree as follows:
$$
c(G_{\tau^*}) = \sum_{e \in E_{\tau^*}} \sum_{\pi: e \in E_{\pi}}
y_{\pi,\tau^*} = \sum_{\pi \in \Pi} \left|E_{\tau^*} \cap
E_{\pi}\right|\cdot y_{\pi,\tau^*}.
$$
By construction the set $E_{\tau^*} \cap E_{\pi}$ has cardinality
exactly $r(\pi)-1$ for all $\pi \in \Pi$ with $y_{\pi,\tau^*}>0$. We
obtain that $ \sum_{e \in E} c_e x_{e,\tau^*} = \sum_{\pi \in \Pi}
(r(\pi)-1) \cdot y_{\pi,\tau^*} $ and this finishes the proof of the
lemma. \qquad\end{proof}

Observe that the above primal-dual algorithm is indeed Kruskal's
algorithm: if the algorithm adds an edge $e$ at time $\tau$, then
$e$ is the minimum-cost edge connecting two connected components of
$G_{\tau}$.

\subsection{A new LP relaxation for Steiner trees}\label{sec:assumptions}
In an instance of the Steiner tree problem, a partition $\pi$ of $V$
is defined to be a {\em Steiner partition} when each part of $\pi$
contains at least one terminal. Chopra and Rao~\cite{CR94a}
introduced this notion and proved that, when $x$ is the incidence
vector of a Steiner tree and $\pi$ is a Steiner partition, the
inequality
\begin{equation}
\label{eq:cr94a} \sum_{e \in E_{\pi}} x_e \geq r(\pi)-1.
\end{equation}
holds. These \emph{Steiner partition inequalities} motivate our
approach.

In the following we use $G[U]$ to denote the subgraph of $G$ induced
by vertex set $U$, i.e., the graph with vertex set $U$ and such that
$E(G[U]) = \{uv \in E(G) \mid u \in U, v \in U\}$. We make the
following assumptions:
\begin{enumerate}
\item[A1.] $G[R]$ is a complete graph and, for any two terminals $u,v \in R$,
  $c_{uv}$ is the cost of a minimum-cost $u,v$-path in $G$.
\item[A2.] For every Steiner vertex $v$ and every vertex $u \in S_v \cup R$,
  $uv$ is an edge of $G$, and $c_{uv}$ is the cost of a minimum-cost $u,v$-path in $G$.
\end{enumerate}
It is a well-known fact that these assumptions are w.l.o.g., i.e.,
any given instance can be transformed into an equivalent instance
that satisfies these assumptions (e.g., see \cite{Va01}). Note that
$b$-quasi-bipartiteness is preserved by these assumptions.

Recall from Section \ref{sec:rSteiner} that a full component is a tree whose
internal vertices are Steiner vertices and all of whose leaves are terminals.
Also recall that a full component $K$ is $r$-restricted if it contains
at most $r$ terminals. Further, the edge-set of any $r$-restricted Steiner tree $T$ can be partitioned
into $r$-restricted full components.
From now on, let $r \geq 2$ be an arbitrary fixed constant.
Define
$$\fC_r := \{ K \subseteq R \,:\, 2 \leq |K| \leq r \textrm{ and there
exists a full component whose terminal set is }K\}.$$ We note that,
for each $K \in \fC_r$, we can determine a minimum-cost full
component with terminal set $K$ in polynomial time (e.g., by using
the dynamic programming algorithm of Dreyfus and
Wagner~\cite{DW72}). Thus, we can compute $\fC_r$ in polynomial time
as well.

For brevity we will abuse notation slightly and use $K \in \fC_r$
interchangeably for a subset of the terminal set and for a particular
min-cost full component spanning $K$. Given any $r$-restricted Steiner
tree, we may assume that all of its full components are from $\fC_r$,
without increasing its cost.

For each full component $K$, we use $E(K)$ to denote its edges,
$V(K)$ to denote its vertices (including Steiner vertices), and
$c_K$ to denote its cost. For a set $\ess$ of full components we
define $E(\ess) := \cup_{K \in \ess}E(K)$ and similarly $V(\ess) :=
\cup_{K \in \ess}V(K)$. By assumption A1 we may assume that the full
component for a terminal pair is just the edge linking those
terminals, and by assumption A2 we may assume that any Steiner node
has degree at least 3. We will also assume that any two distinct
full components $K_1, K_2 \in \fC_r$ are edge disjoint and
internally vertex disjoint. This assumption is without loss of
generality as each Steiner vertex in $G$ can be cloned a sufficient
number of times to ensure this property. Finally, we redefine $G$ to
be $(V(\K), E(\K))$; as a result, the Steiner trees of the new graph
correspond to the $r$-restricted Steiner trees of the original
graph.

Let $\fC_r(T)$ denote the set of all full components of a Steiner tree
$T$. For an arbitrary subfamily $\ess$ of the full components $\fC_r$,
our new LP uses the following canonical decomposition of a Steiner
tree into elements of $E(\ess)$ and $\fC_r \bs \ess$. The idea, as we
will explain later, is to iteratively select a ``good" set $\ess$.

\begin{definition} \label{def:sdecomp}
If $T$ is an $r$-restricted Steiner tree, its {\em
$\ess$-decomposition} is the pair
$$(E(T) \cap E(\ess), \K(T) \bs \ess).$$
\end{definition}

Observe that after $\ess$-decomposing a Steiner tree $T$ we have
$$ \sum_{e \in E(T) \cap E(\ess)} c_e + \sum_{K \in \K(T) \bs \ess} c_K = c(T).$$
We hence obtain a new higher-dimensional view of the Steiner tree
polyhedron. Define
\begin{eqnarray*}
 \ST^\ess  := \mathtt{conv} \{x \in \{0,1\}^{E(\ess)} \times \{0,1\}^{\fC_r\bs\ess} & \,:\, &
            \exists T \in \ST\mbox{ s.t. } x \mbox{ is the incidence} \\
       &  & \mbox{ vector of the $\ess$-decomposition of } T\}.
\end{eqnarray*}

\ignore{
We want to formulate an LP relaxation of the $r$-restricted Steiner
tree problem that makes use of this decomposition, and we do so
using partition inequalities. It is based on the following simpler
idea from \cite{CR94a,CR94b}, which is itself a generalization of
\cite{C89}.
}

The following definitions are used to generalize Steiner partition
inequalities to use full components. We use $\Pi^\ess$ to denote the
family of all partitions of $V(\ess) \cup R$.

\ignore{ Merging the two parts of $\pi$ containing the endpoints of
an edge $e \in E_{\pi}(\ess)$ yields a partition of rank
$r(\pi)-1=p-1$. In other words, adding edge $e$ to a Steiner tree
contributes $1$ to the rank requirement of $\pi$.

Now consider a full component $K \in \fC_r \bs {\ess}$. How much
does $K$ contribute to the rank requirement of $\pi$? In other
words, if we merged those parts of $\pi$ that contain terminals
spanned by $K$, what would the rank of the resulting partition be?
The following definition defines the rank drop in $\pi$ incurred by
adding $K$. In the rest of the paper, we define $\sppart[\ess]$ to
be the set of all partitions of the vertex set $R\cup V(\ess)$.
\ignore{for any $\ess \subseteq V\setminus R$. } }
\begin{definition}
Let $\pi=\{V_1,\ldots,V_p\} \in \sppart[\ess]$ be a partition of the
set $R\cup V(\ess)$. The {\em rank contribution} of full component
$K \in \fC_r \bs {\ess}$ is defined as
$$ \rc^{\pi}_K := |\{i \,:\, K\mbox{ contains a terminal in } V_i\}| - 1. $$
The {\em Steiner rank} $\rs(\pi)$ of $\pi$ is defined as
$$ \rs(\pi) := \{\textrm{the number of parts of $\pi$ that contain terminals}\}.$$
\end{definition}

We describe below a new LP relaxation \eqref{p-st} of $\ST^\ess$.
The relaxation has a variable $x_e$ for each $e \in E(\ess)$ and a
variable $x_K$ for each $K \in \fC_r \bs {\ess}$. For a partition
$\pi \in \Pi^\ess$, we define $E_\pi(\ess)$ to be the edges of
$\ess$ whose endpoints lie in different parts of $\pi$, i.e.,
$E_\pi(\ess) = E(\ess) \cap E_\pi$.

\begin{align}
\min \quad & \sum_{e \in E(\ess)} c_e \cdot x_e + \sum_{K \in \K \bs
\ess} c_K \cdot x_K \tag{$\mathrm{P}_{ST}^\ess$}
\label{p-st} \\
\mbox{s.t} \quad  & \sum_{e \in E_{\pi}(\ess)} x_e + \sum_{K \in
\fC_r \bs {\ess}} \rc^{\pi}_K\cdot x_K \geq \rs(\pi)-1 & &
\forall \pi \in \sppart[\ess]  \label{eq:p}\\
& x_e,x_K \geq 0 & & \forall e \in E(\ess), K \in \fC_r \bs {\ess}
\label{p-st-nn}
\end{align}

Its LP dual has a variable $y_{\pi}$ for each partition $\pi \in
\sppart[\ess]:$
\begin{align}
\max \quad & \sum_{\pi \in \sppart[\ess]} (\rs(\pi)-1)\cdot y_{\pi}
 \tag{$\textrm{D}_{ST}^\ess$} \label{d-st} \\
\mbox{s.t} \quad  & \sum_{\pi \in \sppart[\ess]: e \in
E_{\pi}(\ess)} y_{\pi} \leq c_e
&& \forall e \in E  \label{d-st:1} \\
& \sum_{\pi \in \sppart[\ess]} \rc^{\pi}_K \cdot y_{\pi}\leq c_K
&& \forall K \in \fC_r \bs {\ess}  \label{d-st:2}\\
& y_{\pi} \geq 0,
& & \forall \pi \in \sppart[\ess]
\label{d-st-nn}
\end{align}

We conclude this section with a proof that the (primal) LP is indeed
a relaxation of the convex hull of $\ess$-decompositions for
$r$-restricted Steiner trees. Obviously, constraints
\eqref{p-st-nn} 
hold whenever $x$ is the incidence vector of the
$\ess$-decomposition of a Steiner tree.

\begin{lemma}\label{lem:genpart}
The inequality \eqref{eq:p} is valid for $\ST^\ess$.
\end{lemma}
\begin{proof}
Suppose, for the sake of contradiction, that \eqref{eq:p} is not
valid for $\ST^\ess$ for this $\pi$. Then there must exist a
feasible Steiner tree $T$ with $\ess$-decomposition $(E(T) \cap
E(\ess),\K(T) \bs \ess)$ whose incidence vector $x \in \ST^\ess$
violates \eqref{eq:p} for some partition $\pi \in \sppart[\ess]$.
Choose such a partition $\pi$ with smallest rank.

Observe first that $\pi$ must be a Steiner partition. Otherwise,
there is a part $V_1$ of $\pi$ that contains no terminals. Let $V_2$
be a part in $\pi$ that contains terminals and obtain a new
partition $\pi'$ from $\pi$ by merging $V_1$ and $V_2$.  As $V_1$
contains no terminals, we clearly have $\rc^{\pi}_K=\rc^{\pi'}_K$
for all full components $K \in \fC_r$. Also, the Steiner rank of
$\pi$ and $\pi'$ is the same. As $e \in E_{\pi'}(\ess)$ implies that
$e \in E_{\pi}(\ess)$, it follows that \eqref{eq:p} is violated for
$\pi'$ as well and $\pi'$ has smaller rank than $\pi$ which
contradicts our choice.

Suppose that $V(T) \subseteq R \cup V(\ess)$. This would mean that
$\fC_r(T) \bs {\ess} =\emptyset$ and in this case, Equation
\eqref{eq:cr94a} implies that
$$ \sum_{e \in E_{\pi}(\ess)} x_e \geq r(\pi)-1. $$
Thus, inequality \eqref{eq:p} holds for $\pi$ and $x$ which is a
contradiction.

We may therefore assume that $\fC_r(T) \bs {\ess}$ contains some
full component $\bar{K}$. We obtain a new partition $\pi'$ from
$\pi$ by merging those parts of $\pi$ that contain terminals spanned
by $\bar{K}$. The rank of this new partition is $r(\pi) -
\rc^{\pi}_{\bar{K}}$. It follows from our choice of $\pi$ that
$$ \sum_{e \in E_{\pi'}(\ess)} x_e + \sum_{K \in \fC_r \bs {\ess}} \rc^{\pi'}_K x_K \geq
r(\pi')-1 = r(\pi) - \rc^{\pi}_{\bar{K}} - 1. $$ Now note that
$E_{\pi'}(\ess) \subseteq E_{\pi}(\ess)$ and $\rc_{\bar{K}}^{\pi'} =
0$, and that $\rc^{\pi'}_K \leq \rc^{\pi}_K$ for all $K \in \fC_r
\bs \ess$. The above inequality therefore implies
$$ \sum_{e \in E_{\pi}(\ess)} x_e + \sum_{K \in \fC_r \bs {\ess}} \rc^{\pi}_K x_K \geq
\sum_{e \in E_{\pi'}(\ess)} x_e + \sum_{K \in \fC_r \bs {\ess} \bs
\{\bar{K}\}} \rc^{\pi'}_K x_K + \rc^{\pi}_{\bar{K}} \geq r(\pi) -
\rc^{\pi}_{\bar{K}} - 1 + \rc^{\pi}_{\bar{K}}
$$ which in turn proves that \eqref{eq:p} holds for $\pi$ and $x$.
This contradiction completes the proof of the lemma.
\end{proof}

\section{An iterated primal-dual algorithm for Steiner trees}\label{sec:zel}

\ignore{In this section we present a simple iterated primal-dual algorithm
that captures many of the greedy Steiner tree algorithms
that were introduced in Section \ref{sec:intro}.
In what follows, we fix an
instance of the Steiner tree problem, i.e., we fix the undirected
input-graph $G=(V,E)$, the non-negative cost function $c \geq 0$ and
some constant $r \geq 2$.}

As described in Section \ref{sec:kruskal}, $\MST(G,c)$ denotes a
call to Kruskal's minimum-spanning tree algorithm on graph $G$ with
cost-function $c$. It returns a minimum-cost spanning tree $T$ and
an optimal feasible dual solution $y$ for \eqref{d-sp}. Let $\mst(G,
c)$ denote the cost of $\MST(G, c)$. Since $c$ is fixed, in the rest
of the paper we omit $c$ where possible for brevity. Let us also
abuse notation and identify each set $\ess \subset \K$ of full
components with the graph $(V(\ess), E(\ess))$.

\ignore{Let $T^*$ be a Steiner tree of minimum cost for the given
instance and let $S^*$ be the set of Steiner vertices in $T^*$,
i.e., $S^*=V(T^*)\setminus R$. Then $T^*$ is a minimum-cost spanning
tree in the graph $G[R\cup S^*]$ that is induced by the vertices in
$R\cup S^*$.  Thus, identifying the set $S^*$ of Steiner vertices is
an \NP-hard problem by itself.}

The main idea of the greedy algorithms in \cite{RZ05,Ze96,Ze93} is
to find a set $\ess \subset \K$ of full components such that
$\MST(\ISC{\ess})$ has small cost relative to $\opt_r$. Let
$\tbinom{R}{2}$ denote the collection of all pairs of terminals. The
algorithms all start with $\ess = \tbinom{R}{2}$ and then grow
$\ess$, so for the rest of the paper we assume that
$\tbinom{R}{2}\subseteq \ess$; hence $E(G[R]) \subseteq E(\ess)$ and
$R \subseteq V(\ess)$.

The reason that $\MST$ is useful in our primal-dual framework is
that we can relate the dual program \eqref{d-sp} on graph $\ess$ to
the dual program \eqref{d-st}. Let $y$ be the feasible dual returned
by a call to $\MST(\ISC{\ess})$. We treat $y$ as a dual solution of
\eqref{d-st} by setting each $y_K$ to zero; note that constraints
\eqref{d-sp:edge} and \eqref{d-sp-nn} of \eqref{d-sp} imply that $y$
also meets constraints \eqref{d-st:1} and \eqref{d-st-nn} of
\eqref{d-st}.
If $K$ is a full component
such that \eqref{d-st:2} does not hold for $y$, we say that
$K$ is \emph{violated} by $y$.

The primal-dual algorithm finds such a set $\ess$ in an iterative
fashion. Initially, $\ess$ is equal to $\tbinom{R}{2}$. In each
iteration, we compute a minimum-cost spanning tree $T$ of the graph
$\ISC\ess$. The dual solution $y$ corresponding to this tree is
converted to a dual for \eqref{d-st}, and if $y$ is feasible for
\eqref{d-st}, we stop. Otherwise, we add a violated full component
to $\ess$ and continue. The algorithm clearly terminates (as $\K$ is
finite) and at termination, it returns the final tree $T$ as an
approximately-optimal Steiner tree.

Algorithm \ref{alg} summarizes the above description. The greedy
algorithms in \cite{RZ05,Ze96,Ze93} differ only in how $K$ is
selected in each iteration, i.e., in how the selection function
$f_i:\fC_r \rightarrow \IR$ is defined (see also \cite[\S
1.4]{GH+01b} for a well-written comparison of these algorithms).

\begin{algorithm}
\caption{\label{alg} A general iterative primal-dual framework for Steiner
  trees.}
\begin{algorithmic}[1]
\STATE Given: Undirected graph $G=(V,E)$, non-negative costs
       $c_e$ for all edges $e \in E$, constant $r\geq 2$.
\STATE $\ess^0:=\tbinom{R}{2}$, $i:=0$
\REPEAT
  \STATE $(T^i,y^i) := \MST(\ISC{\ess^i})$
  \IF{$y^{i}$ is not feasible for $\REF{\textrm{D}_{ST}^{\ess^{i}}}$}
  \STATE Choose a violated full component $K^i \in \fC_r \bs {\ess^i}$ such
       that $f_i(K^i)$ is minimized
\STATE $\ess^{i+1} := \ess^i \cup \{K^i\}$ \ENDIF \STATE $i:=i+1$
\UNTIL{$y^{i-1}$ is feasible for $\REF{\textrm{D}_{ST}^{\ess^{i-1}}}$}
\STATE Let $p=i-1$ and return $(T^p,y^p)$.
\end{algorithmic}
\end{algorithm}

The following lemma is at the heart of our proof, and explains why
our LP can be used to find cheap Steiner trees.

\begin{lemma}\label{lem:feas}
Let $(T, y) = \MST(\ISC{\ess})$ and suppose that $K$ is violated by
$y$. Then adding $K$ to $\ess$ produces a cheaper spanning
tree, i.e.,
$$ \mst(\ISC{\ess \cup \{K\}}) < c(T). $$
\end{lemma}
\begin{proof}Assume that $\MST(\ISC{\ess})$ finishes at time $\tau^*$ and,
once again, let $\pi_{\tau}$ be the partition maintained by
Kruskal's algorithm at time $0 \leq \tau \leq \tau^*$.

Define $q=\rc_K^{\pi_0}$ to be the rank-contribution of $K$ with
respect to the initial partition. Clearly, $\rc_K^{\pi_{\tau^*}}=0$
as all terminals are contained in the same connected component at
time $\tau^*$.  Then there are edges $e_1, \ldots, e_q \in T$ such
that, for $1 \leq i \leq q$, the rank-contribution of $K$ with
respect to the partition maintained by Kruskal's algorithm drops
from $q-i+1$ to $q-i$ when edge $e_i$ is added. Formally, for $1
\leq i \leq q$, let $\pi_i$ and $\pi_i'$ be the partition maintained
by Kruskal's algorithm before and after adding edge $e_i$, then
$$ \rc_K^{\pi_i} = \rc_K^{\pi_i'} + 1. $$
We denote the time of addition of edge $e_i$ by $\tau_i$ for all
$i$.

From the description of Kruskal's algorithm it follows that
$$ \sum_{i=1}^q c_{e_i} = \sum_{i=1}^q \tau_i = \int_{0}^{\tau^*}
\rc^{\pi_{\tau}}_K d\tau $$ and the right-hand side of this equality
is equal to $\sum_{\pi \in \sppart[\ess]} \rc^{\pi}_K y_{\pi}$. The
fact that constraint \eqref{d-st:2} is violated for $K$ therefore
implies that
$$ c_{e_1} + \dotsb + c_{e_q} > c_K.$$
Finally observe that $T \cup E(K) \setminus \{e_1, \ldots, e_q\}$ is
a spanning tree of $\IS{\ess \cup \{K\}}$ and its cost is smaller
than that of $T$.
\end{proof}

\subsection{Cutting losses: the {\tt RZ} selection function}

A potential weak point in Algorithm \ref{alg} is that once a full
component is added to $\ess$, it is never removed. On the other hand,
if some cheap subgraph $H$ connects all Steiner vertices of $\ess$ to
terminals, then adding $H$ to any Steiner tree gives us a tree that
spans $V(\ess)$, i.e., we have so far {\it lost} at most $c(H)$ in the
final answer. This leads to the concept of the {\em loss} of a Steiner
tree which was first introduced by Karpinski and Zelikovsky in
\cite{KZ97}.

\begin{definition}
Let $G'=(V',E')$ be a subgraph of $G$. The {\em loss} $\LOSS(G')$ is
a minimum-cost set $E'' \subseteq E'$ such that every connected
component of $(V',E'')$ contains a terminal. Let $\loss(G')$ denote
the cost of $\LOSS(G')$.
\end{definition}

See Figure \ref{fig:loss} for an example of the loss of a graph. The
above discussion amounts to saying that $\min\{\mst(\ISC{\ess'})
\mid \ess' \supseteq \ess\} \leq \opt_r + \loss(\ess)$.
Consequently, our selection function $f_i$ in step 6 of the
algorithm should try to keep the loss small. The following fact
holds because full components in $\K$ meet only at terminals.
\begin{fact}\label{fact:modloss}
If $\mathcal{S} \subseteq \K$, then $\LOSS(\IS{\mathcal{S}}) =
\cup_{K \in \mathcal{S}} \LOSS(K)$ and so $\loss(\IS{\mathcal{S}}) =
\sum_{K \in \mathcal{S}} \loss(K)$.
\end{fact}

For a set $\ess$ of full components, where $y$ is the dual solution
returned by $\MST(\ISC{\ess})$, define
\begin{equation}\label{eq:smst}
 \smst(\ISC{\ess}) := \sum_{\pi \in \Pi^\ess} (\rs(\pi)-1)y_{\pi}.
\end{equation}
If $y$ is feasible for \eqref{d-st} then by weak LP duality,
$\smst(\ISC{\ess})$ provides a lower bound on $\opt_r$. If $y$ is
infeasible for \eqref{d-st}, then which full component should we
add? Robins and Zelikovsky propose minimizing the ratio of the
change in upper bound to the change in potential lower bound
\eqref{eq:smst}. Their selection function $f_i$ is defined by
\begin{equation}\label{def:fi}
  f_i(K) := \frac{\loss(K)}{\smst(\ISC{\ess^i})
           -\smst(\ISC{\ess^i \cup \{K\}})} =   \frac{\loss(\ess^i \cup \{K\}) - \loss(\ess^i)}{\smst(\ISC{\ess^i})
           -\smst(\ISC{\ess^i \cup \{K\}})},
\end{equation}
where the equality uses Fact \ref{fact:modloss}.

\begin{figure}
\begin{center} \leavevmode
\begin{pspicture}(0.8,0.8)(5.6,3.2)

\psset{unit=.8}

\terminal{1,4}{t1}
\terminal{1,1}{t2}
\terminal{3,1}{t3}
\terminal{3,4}{t4}
\terminal{5,4}{t5}
\terminal{5,2}{t6}
\terminal{7,3}{t7}
\terminal{7,1}{t8}

\steiner{2,3}{s1} \terminal{2,2}{s2} \steiner{4,3}{s3}
\steiner{7,2}{s4}

\ncline[linewidth=.1mm]{t1}{s1} \nbput[labelsep=.1]{$4$}
\ncline[linewidth=.1mm]{t2}{s2} \naput[labelsep=.1]{$2$}
\ncline[linewidth=.1mm]{t3}{s2} \nbput[labelsep=.1]{$6$}
\ncline[linewidth=.1mm]{t4}{s1} \nbput[labelsep=.1]{$5$}
\ncline[linewidth=.8mm]{t5}{s3} \nbput[labelsep=.1]{$3$}
\ncline[linewidth=.1mm]{t6}{s3} \nbput[labelsep=.1]{$8$}
\ncline[linewidth=.1mm]{t7}{s4} \naput[labelsep=.1]{$5$}
\ncline[linewidth=.8mm]{t8}{s4} \nbput[labelsep=.1]{$3$}
\ncline[linewidth=.1mm]{t6}{s4} \naput[labelsep=.1]{$4$}
\ncline[linewidth=.1mm]{s1}{s2} \nbput[labelsep=.1]{$5$}
\ncline[linewidth=.8mm]{s1}{s3} \naput[labelsep=.1]{$2$}

\end{pspicture}
\caption{\label{fig:loss} The figure shows the Steiner tree instance
from Figure \ref{fig:decomp} with costs on the edges. The loss of
the Steiner tree in this figure is shown in thick edges. Its cost is
$8$.}
\end{center}
\end{figure}
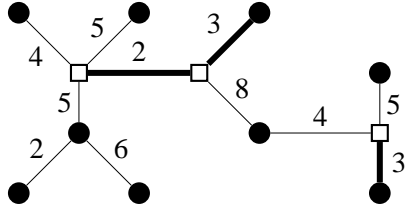

\section{Analysis}\label{sec:ana}
Fix an optimum $r$-Steiner tree $T^*$. There are several steps in
proving the performance guarantee of Robins and Zelikovsky's
algorithm, and they are encapsulated in the following result, whose
complete proof appears in Section \ref{sec:tpcost}.


\begin{lemma}\label{lem:tpcost}
The cost of the tree $T^p$ returned by Algorithm \ref{alg} is at
most
$$ \opt_r + \loss(T^*) \cdot \ln\left(1+\frac{\smst(G[R],c)
  - \opt_r}{\loss(T^*)}\right). $$
\end{lemma}

The main observation in the proof of the above lemma
can be summarized as follows: from the
discussion in Section \ref{sec:mst}, we know that the tree $T^p$
returned by Algorithm \ref{alg} has cost
$$ \mst(\ISC{\ess^p}) = \sum_{\pi \in \sppart[\ess^p]} (r(\pi)-1)y^p_{\pi} $$
and the corresponding lower-bound on $\opt_r$ returned by the
algorithm is
$$ \smst(\ISC{\ess^p}) = \sum_{\pi \in \sppart[\ess^p]} (\rs(\pi)-1)y^p_{\pi}. $$
We know that $\smst(\ISC{\ess^p}) \leq \opt_r$ but how large is
the difference between $\mst(\ISC{\ess^p})$ and $\smst(\ISC{\ess^p})$? We show that the difference
$$ \sum_{\pi \in \sppart[\ess^p]} (r(\pi)-\rs(\pi)) y^p_{\pi} $$
is exactly equal to the loss $\loss(T^p)$ of tree $T^p$. We then
bound the loss of each selected full component $K^i$, and putting
everything together finally yields Lemma \ref{lem:tpcost}.

The following lemma states the performance guarantee of Moore's
minimum-spanning tree heuristic as a function of the optimum loss
and the maximum cardinality $b$ of any Steiner neighborhood in $G$.

\begin{lemma}\label{lem:bquasi}
Fix an arbitrary optimum $r$-restricted Steiner tree $T^*$. Given an
undirected, $b$-quasi-bipartite graph $G=(V,E)$, a set of terminals
$R \subseteq V$, and non-negative costs $c_e$ for all $e \in E$, we
have
$$ \mst(G[R],c) \leq 2\opt_r - \frac{2}{b}\loss(T^*) $$
for any $b\geq 1$.
\end{lemma}
\begin{proof}
Recall that $\K(T^*)$ is the set of full components of tree $T^*$.
Now consider a full component $K \in\K(T^*)$. We will now show that
there is a minimum-cost spanning tree of $G[K]$ whose cost is at
most $2c_K -\frac{2}{b}\loss(K)$. By repeating this argument for all
full components $K \in \K(T^*)$, adding the resulting bounds, and
applying Fact \ref{fact:modloss}, we obtain the lemma.

For terminals $r, s \in K$, let $P_{rs}$ denote the unique
$r,s$-path in $K$. Pick $u, v \in K$ such that $c(P_{uv})$ is
maximal. Define the \emph{diameter}
 $\Delta(K) := c(P_{uv})$.
Do a depth-first search traversal of $K$ starting in $u$ and ending
in $v$. The resulting walk in $K$ traverses each edge not on
$P_{uv}$ twice while each edge on $P_{uv}$ is traversed once. Hence
the walk has cost $2c_K - \Delta(K)$. Using standard short-cutting
arguments it follows that the minimum-cost spanning tree of $G[K]$
has cost at most
\begin{equation}\label{bquasi:1}
2c_K-\Delta(K)
\end{equation}
as well.

Each Steiner vertex $s \in V(K) \bs R$ can connect to some terminal $v \in
K$ at cost at most $\frac{\Delta(K)}{2}$. Hence, the cost $\loss(K)$ of the
loss of $K$ is at most $b\frac{\Delta(K)}{2}$. In other words we have
$\Delta(K) \geq \frac{2}{b}\loss(K)$. Plugging this into
\eqref{bquasi:1} yields the lemma.
\end{proof}

For small values of $b$ we can obtain additional improvements via
case analysis.
\begin{lemma}\label{lem:bquasi2}
Suppose $b \in \{3, 4\}$. Fix an arbitrary optimum $r$-restricted
Steiner tree $T^*$. Given an undirected, $b$-quasi-bipartite graph
$G=(V,E)$, a set of terminals $R \subseteq V$, and non-negative
costs $c_e$ for all $e \in E$, we have
$$ \mst(G[R],c) \leq 2\opt_r - \loss(T^*). $$
\end{lemma}
\begin{proof}
As in the proof of Lemma \ref{lem:bquasi} it suffices to prove that,
for each full component $K \in \K(T^*)$, there is a minimum-cost
spanning tree of $G[K]$ whose cost is at most $2c_K - \loss(K)$, for
then we can add the bound over all such $K$ to get the desired
result. For terminals $r, s \in K$, let $P_{rs}$ again denote the
unique $r,s$-path in $K$.

Notice that the Steiner nodes (there are at most $b$ of them) in the
full component $K$ either form a path, or else there are 4 of them
and they form a star.

\begin{description}
\item[Case 1:] the Steiner nodes in $K$ form a path. Let $x$ and $y$
be the Steiner nodes on the ends of this path. Let $u$ (resp.\ $v$)
be any terminal neighbour of $x$ (resp.\ $y$); see Figure
\ref{fig:case}(i) for an example. Perform a depth-first search in
$K$ starting from $u$ and ending at $v;$ the cost of this search is
$2c_K-c(P_{uv})$. By standard short-cutting arguments it follows
that $2c_K-c(P_{uv})$ is an upper bound on $\mst(G[K])$. On the
other hand, since $P_{uv} \bs \{ux\}$ is a candidate for the loss of
$K$, we know that $\loss(K) \leq c(P_{uv} \bs \{ux\}) \leq
c(P_{uv})$. Therefore we obtain
\begin{equation}\label{eq:losss}\mst(G[K]) \leq 2c_K-c(P_{uv}) \leq
2c_K-\loss(K).\end{equation}

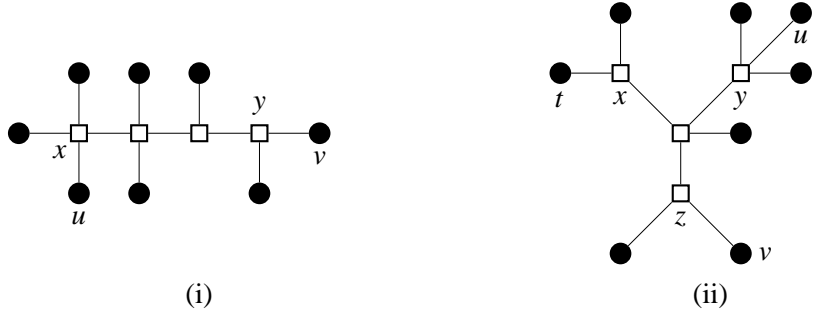
\begin{figure}
\begin{center} \leavevmode
\begin{pspicture}(0.8,0)(12.8,4.4)

\psset{unit=.8}

\steiner{2,3}{ls1} \steiner{3,3}{ls2} \steiner{4,3}{ls3}
\steiner{5,3}{ls4} \terminal{1,3}{lt1a} \terminal{2,2}{lt1b}
\terminal{2,4}{lt1c} \terminal{3,2}{lt2a} \terminal{3,4}{lt2b}
\terminal{4,4}{lt3} \terminal{6,3}{lt4a} \terminal{5,2}{lt4b}

\uput{.3}[225](2,3){$x$} \uput{.3}[90](5,3){$y$}
\uput{.3}[270](2,2){$u$} \uput{.3}[270](6,3){$v$}

 { \psset{linewidth=.1mm} \ncline{ls1}{ls2}
\ncline{ls2}{ls3} \ncline{ls3}{ls4} \ncline{ls1}{lt1a}
\ncline{ls1}{lt1b} \ncline{ls1}{lt1c} \ncline{ls2}{lt2a}
\ncline{ls2}{lt2b} \ncline{ls3}{lt3} \ncline{ls4}{lt4a}
\ncline{ls4}{lt4b} }

\uput{.5}[270](4,1){(i)}

\steiner{12,3}{rs1} \steiner{11,4}{rs2} \steiner{13,4}{rs3}
\steiner{12,2}{rs4} \terminal{13,3}{rt1} \terminal{10,4}{rt2a}
\terminal{11,5}{rt2b} \terminal{13,5}{rt3a} \terminal{14,5}{rt3b}
\terminal{14,4}{rt3c} \terminal{11,1}{rt4a} \terminal{13,1}{rt4b}

\uput{.3}[270](11,4){$x$} \uput{.3}[270](13,4){$y$}
\uput{.3}[270](12,2){$z$} \uput{.3}[270](10,4){$t$}
\uput{.3}[270](14,5){$u$} \uput{.3}[0](13,1){$v$}

 { \psset{linewidth=.1mm} \ncline{rs1}{rs2}
\ncline{rs1}{rs3} \ncline{rs1}{rs4} \ncline{rs1}{rt1}
\ncline{rs2}{rt2a} \ncline{rs2}{rt2b} \ncline{rs3}{rt3a}
\ncline{rs3}{rt3b} \ncline{rs3}{rt3c} \ncline{rs4}{rt4a}
\ncline{rs4}{rt4b} }

\uput{.5}[270](12.5,1){(ii)}

\end{pspicture}
\caption{\label{fig:case} The figure shows the two types of full
components when $b \leq 4$. On the left is a full component where
the Steiner nodes form a path, and on the right is a full component
where the Steiner nodes form a star with 3 tips.}
\end{center}
\end{figure}

\item[Case 2:] the Steiner nodes in $K$ form a star. Let the
tips of the star be $x, y, z$ and let $t, u, v$ be any terminal
neighbours of $x, y, z$ respectively; see Figure \ref{fig:case}(ii)
for an example. Without loss of generality, we may assume that
$c_{xt} \leq c_{yu} \leq c_{zv}$. As before, a depth-first search in
$K$ starting from $u$ and ending at $v$ has cost $2c_K-c(P_{uv})$
and this is an upper bound on $\mst(G[K])$. On the other hand,
$P_{uv} \bs \{yu\} \cup \{xt\}$ is a candidate for the loss of $K$
and so $\loss(K) \leq c(P_{uv}) - c_{yu} + c_{xt} \leq c(P_{uv})$.
We hence obtain Equation \eqref{eq:losss} as in the previous case.
\qedhere
\end{description}\end{proof}

We are ready to prove our main theorem. We restate it using the
notation introduced in the last two sections.

\setcounter{ntheorem}{\value{theorem}}
\setcounter{theorem}{\value{nntheorem}}
\begin{theorem}
Given an undirected, $b$-quasi-bipartite graph $G=(V,E)$, terminals
$R \subseteq V$, and a fixed constant $r \geq 2$, Algorithm
\ref{alg} returns a feasible Steiner tree $T^p$ with
$$
  c(T^p) \leq \left\{\begin{array}{l@{\quad : \quad}l}
      1.279 \cdot \opt_r & b=1 \\
      (1+1/e) \cdot \opt_r & b \in \{2,3,4\} \\
      \left(1+\frac{1}{2}\ln\left(3-\frac2b\right)\right)\opt_r & b \geq 5.
      \end{array}\right.
$$
\end{theorem}
\begin{proof}
Using Lemma \ref{lem:tpcost} we see that
\begin{eqnarray}
  c(T^p) & \leq & \opt_r + \loss(T^*) \cdot \ln\left(1+\frac{\smst(G[R],c)
  - \opt_r}{\loss(T^*)}\right) \notag \\
  & = & \opt_r + \loss(T^*) \cdot \ln\left(1+\frac{\mst(G[R],c)
  - \opt_r}{\loss(T^*)}\right). \label{donut}
\end{eqnarray}
The second equality above holds because $G[R]$ has no Steiner
vertices. Applying the bound on $\mst(G[R],c)$ from Lemma
\ref{lem:bquasi} yields
\begin{equation}\label{main:1}
  c(T^p) \leq \opt_r \cdot \left[ 1 + \frac{\loss(T^*)}{\opt_r} \cdot
  \ln\left(1-\frac{2}{b}+\frac{\opt_r}{\loss(T^*)} \right)\right].
\end{equation}
Karpinski and Zelikovsky~\cite{KZ97} show that $\loss(T^*) \leq
\frac{1}{2}\opt_r$. We can therefore obtain an upper-bound on the
right-hand side of \eqref{main:1} by bounding the maximum value of
function $x\ln(1-2/b+1/x)$ for $x \in [0,1/2]$. We branch into
cases:
\begin{description}
\item[$b=1$:] The maximum of $x\ln(1/x - 1)$ for $x \in [0,1/2]$ is attained
  for $x \approx 0.2178$. Hence, $x\ln(1/x-1) \leq 0.279$ for $x \in [0,1/2]$.
\item[$b=2$:] The maximum of $x\ln(1/x)$ is attained for $x=1/e$ and hence
  $x\ln(1/x) \leq 1/e$ for $x \in [0,1/2]$.
\item[$b \in \{3, 4\}$:] We use Equation~\eqref{donut} together with Lemma \ref{lem:bquasi2}
in place of Lemma \ref{lem:bquasi}; the subsequent analysis and
result are the same as in the previous case.
\item[$b\geq 5$:] The function $x\ln(1-2/b+1/x)$ is
  increasing in $x$ and its maximum is attained for $x=1/2$. Thus,
  $x\ln(1-2/b+1/x) \leq \frac{1}{2} \ln(3-2/b)$ for $x \in [0,1/2]$.
\end{description}
The three cases above conclude the proof of the theorem.
\end{proof}
\setcounter{theorem}{\value{ntheorem}}

\section{Properties of \ensuremath{\mathrm{(P}_{ST}^{\ess}\mathrm{)}}}\label{sec:gap}
In this section, we first prove that the linear program
\ensuremath{\mathrm{(P}_{ST}^{\ess}\mathrm{)}} is gradually weakened
as the algorithm progresses (i.e., as more full components are added
to $\ess)$. Then we describe bounds on the integrality gap of the
new LP, and its strength compared to other LPs for the Steiner tree
problem.

\begin{lemma}\label{lemma:gap++}
If $\ess \subset \ess'$, then the integrality gap of \eqref{p-st} is
at most the integrality gap of \foob.
\end{lemma}
\begin{proof}
We consider only the case where $\ess' = \ess \cup\{J\}$ for some
full component $J;$ the general case then follows by induction on
$|\ess' \bs \ess|$.

Let $x$ be any feasible primal point for \eqref{p-st} and define the
\emph{extension} $x'$ of $x$ to be a primal point of \foob, with
$x'_e = x_{J}$ for all $e \in E(J)$ and $x'_Z = x_Z$ for all $Z \in
(\K \bs \ess') \cup E(\ess)$. We claim that $x'$ is feasible for
\foob. Since $x$ and $x'$ have the same objective value, this will
prove Lemma \ref{lemma:gap++}.

It is clear that $x'$ satisfies constraints
\eqref{p-st-nn}, 
so now let us show that $x'$ satisfies the partition inequality
\eqref{eq:p} in \foob. Fix an arbitrary partition $\pi'$ of $
V(\ess')$, and let $\pi$ be the restriction of $\pi'$ to $ V(\ess)$.
We get
\begin{equation}\label{pizza}\sum_{e \in E_{\pi'}(\ess')} x'_e + \sum_{K \in \fC_r \bs
{\ess'}} \rc^{\pi'}_Kx'_K = \left(\sum_{e \in E_{\pi}(\ess)} x_e +
\sum_{K \in \fC_r \bs {\ess}} \rc^{\pi}_Kx_K\right) + |E_{\pi'}\cap
E(J)|x_J - \rc_{J}^{\pi}x_J.\end{equation} Now $J$ spans at least
$\rc_{J}^{\pi}+1$ parts of $\pi'$, and it follows that
$|E_{\pi'}\cap E(J)| \geq \rc_{J}^{\pi}$. Hence, using Equation
\eqref{pizza}, the fact that $x$ satisfies constraint \eqref{eq:p}
for $\pi$, and the fact that $\rs(\pi) =\rs(\pi')$, we have
$$\sum_{e \in E_{\pi'}(\ess')} x'_e + \sum_{K \in \fC_r \bs {\ess'}}
\rc^{\pi'}_Kx'_K \geq \sum_{e \in E_{\pi}(\ess)} x_e + \sum_{K \in
\fC_r \bs {\ess}} \rc^{\pi}_Kx_K \geq \rs(\pi)-1 = \rs(\pi')-1.$$ So
$x'$ satisfies \eqref{eq:p} for $\pi'$.
\end{proof}

In 1997, Warme \cite{War97} introduced a new linear program for the
Steiner tree problem. He observed (as did the authors of \cite{PS00}
in the same year) that full components allow a reduction from the
Steiner tree problem to the \emph{spanning-tree-in-hypergraph}
problem. He also gave an LP relaxation for spanning trees in
hypergraphs. That LP turns out to be exactly as strong as our own
LP; see \cite[Corollary 3.19]{KP07a} for a proof. Now, Polzin et
al.~\cite{PVd03} proved that Warme's relaxation is stronger than the
bidirected cut relaxation, and Goemans~\cite{Go94} proved that the
(graph) Steiner partition inequalities are valid for the bidirected
cut formulation. Hence, as stated previously, using full components
as in \eqref{p-st} strengthens the Steiner partition inequalities.

\subsection{A lower bound on the integrality gap of \newLP}
Note that when $\ess = \tbinom{R}{2}$,
\newLPsp and \eqref{p-st} are equivalent LPs: for each
terminal-terminal edge $uv$, the full component variable $x_{\{u,
v\}}$ of the former corresponds to the edge variable $x_{uv}$ of the
latter. Hence although we consider the simpler LP \newLP\ in this
section, the results apply also to the LP used in the first
iteration of \rz.

\label{sec:gap-lower-bound} Goemans \cite{AC04} gave a family of
graphs upon which, in the limit, the integrality gap of the
bidirected cut relaxation is $\frac{8}{7}$. Interestingly, it can be
shown that once you preprocess these graphs as described in Section
\ref{sec:assumptions}, the gap completely disappears. Here we
describe another example, due to Skutella \cite{Sk06}. It shows not
only that the gap of the bidirected cut relaxation is at least
$\frac{8}{7}$, but that the gap of our new formulation (including
preprocessing) is at least $\frac{8}{7}$. The example is
quasi-bipartite.

The Fano design is a well-known finite geometry consisting of 7
\emph{points} and 7 \emph{lines}, such that every point is on 3
lines, every line contains 3 points, any two lines meet in a unique
point, and any two points lie on a unique common line. We construct
Skutella's example by creating a bipartite graph, with one side
consisting of one node $n_p$ for each point $p$ of the Fano design,
and the other side consisting of one node $n_\ell$ for each line
$\ell$ of the Fano design. Define $n_p$ and $n_\ell$ to be adjacent
in our graph if and only if $p$ does \emph{not} lie on $\ell$. Then
it is easy to see this graph is 4-regular, and that given any two
nodes $n_1, n_2$ from one side, there is a node from the other side
that is adjacent to neither $n_1$ nor $n_2$. Let one side be
terminals, the other side be Steiner nodes, and then attach one
additional terminal to all the Steiner nodes. We illustrate the
resulting graph in Figure \ref{fig:lp-gap}.

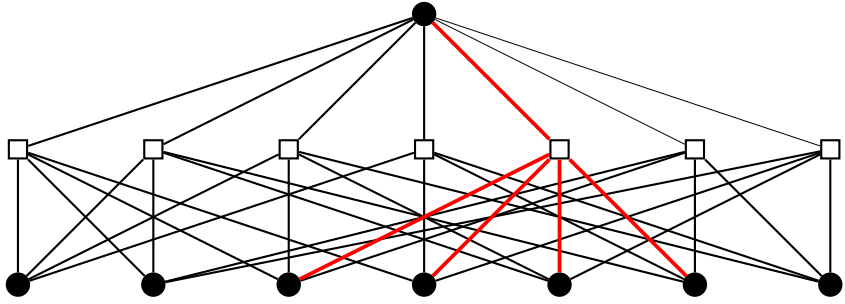
\begin{figure}
\begin{center} \leavevmode
\begin{pspicture}(0.8,0.8)(12.6,4.5)

\psset{unit=.9}

\terminal{1,1}{t1}
\terminal{3,1}{t2}
\terminal{5,1}{t3}
\terminal{7,1}{t4}
\terminal{9,1}{t5}
\terminal{11,1}{t6}
\terminal{13,1}{t7}
\steiner{1,3}{s1}
\steiner{3,3}{s2}
\steiner{5,3}{s3}
\steiner{7,3}{s4}
\steiner{9,3}{s7}
\steiner{11,3}{s5}
\steiner{13,3}{s6}
\terminal{7,5}{root}

\ncline[linewidth=.1mm]{s5}{root}
\ncline[linewidth=.1mm]{s6}{root}

\ncline{s1}{root}
\ncline{s2}{root}
\ncline{s3}{root}
\ncline{s4}{root}
\ncline{t1}{s1}
\ncline{t1}{s2}
\ncline{t1}{s3}
\ncline{t1}{s4}

\ncline{t2}{s1}
\ncline{t2}{s2}
\ncline{t2}{s5}
\ncline{t2}{s6}
\ncline{t3}{s1}
\ncline{t3}{s3}
\ncline{t3}{s5}
\ncline{t4}{s1}
\ncline{t4}{s4}
\ncline{t4}{s6}
\ncline{t5}{s2}
\ncline{t5}{s3}
\ncline{t5}{s6}
\ncline{t6}{s2}
\ncline{t6}{s4}
\ncline{t6}{s5}
\ncline{t7}{s3}
\ncline{t7}{s4}
\ncline{t7}{s5}
\ncline{t7}{s6}

\ncline[linecolor=red,linewidth=.5mm]{t3}{s7}
\ncline[linecolor=red,linewidth=.5mm]{t4}{s7}
\ncline[linecolor=red,linewidth=.5mm]{t5}{s7}
\ncline[linecolor=red,linewidth=.5mm]{t6}{s7}
\ncline[linecolor=red,linewidth=.5mm]{s7}{root}

\end{pspicture}
\end{center}
\caption{\label{fig:lp-gap} Skutella's example, which shows that the
bidirected cut formulation and our new formulation both have a gap
of at least $\frac{8}{7}$. The shaded edges denote one of the
quasi-bipartite full components on 5 terminals.
 \ignore{The dashed edges indicate a flow of value
1 from the root to another terminal.} }
\end{figure}

Each Steiner node is in a unique 5-terminal quasi-bipartite full
component. There are 7 such full components. Denote the family of
these 7 full components by $\cee$.

\begin{claim}
Let $x^*_K = \frac{1}{4}$ for each $K \in \cee$, and $x^*_K = 0$
otherwise. Then $x^*$ is feasible for
\newLP.
\end{claim}
\begin{proof}
It is immediate that $x^*$ satisfies constraints \eqref{p-st-nn}. It
remains only to show that $x^*$ meets constraint \eqref{eq:p}. Let
$\pi$ be an arbitrary partition, with parts $\pi_0, \dotsc, \pi_m$
such that $\pi_0$ contains the extra ``top'' terminal. If we can
show that $\sum_K x^*_K \rc_K^\pi \geq m$ then we will be done,
since $\pi$ was arbitrary. For each $i=1, \dotsc, m$, let $r_i$ be
any terminal in $\pi_i$. Note that each $r_i$ lies in exactly 4 full
components from $\cee$. Furthermore, every full component $K \in
\cee$ satisfies $\rc_K^\pi \geq |K \cap \{r_1, \dotsc, r_m\}|$,
since that full component meets $\pi_0$ as well as each part $\pi_j$
such that $r_j \in K$. Hence \begin{equation}\sum_K x^*_K \rc_K^\pi
= \frac{1}{4} \sum_{K \in \cee} \rc_K^\pi \geq \frac{1}{4} \sum_{K
\in \cee} \# \{j : r_j \in K\} = \frac{1}{4} \sum_{j=1}^m \# \{K \in
\cee : r_j \in K\} = \frac{1}{4} \cdot m \cdot 4 =
m.\tag*{\qedhere}\end{equation}
\end{proof}
The objective value of $x^*$ is $\frac{35}{4}$, but the optimal
integral solution to the LP is 10, since at least 3 Steiner nodes
need to be included. Hence, the gap of our new LP is no better than
$\frac{10}{35/4}=\frac{8}{7}$.

\subsection{A gap upper bound for $b$-quasi-bipartite instances}
\label{sec:1st}

In \cite{RV99} Rajagopalan and Vazirani show that the bidirected cut
relaxation has a gap of at most $\frac{3}{2}$, if the graph is
quasi-bipartite. Since \newLPsp\ is stronger than the bidirected cut
relaxation its gap is also at most $\frac{3}{2}$ for such graphs. We
are able to generalize this result as follows.

\setcounter{gtheorem}{\value{theorem}}
\setcounter{theorem}{\value{ggtheorem}}
\begin{theorem}
On $b$-quasi-bipartite graphs, \newLPsp has an integrality gap
between $\frac{8}{7}$ and $\frac{2b+1}{b+1}$ in the worst case.
\end{theorem}
\setcounter{theorem}{\value{gtheorem}}
\begin{proof}
The lower bound comes from Section \ref{sec:gap-lower-bound}.
We assume $G$ is $b$-quasi-bipartite, we let $T^*$ be an optimal Steiner tree,
and we let $\ess^*$ be its set of full components. Since $T^*$ is a
minimum spanning tree for $\IS{\ess^*}$, there is a corresponding
feasible dual $y$ for \eqref{d-sp}. When we convert $y$ to a dual
for \doob, we claim that $y$ is feasible: indeed, by Lemma
\ref{lem:feas} a violated full component could be used to improve
the solution, but $T^*$ is already optimal. The next lemma is the
cornerstone of our proof.

\begin{lemma}\label{lem:1st:rankdrop}
Let $\pi$ be a partition of $ V(\ess^*)$ with $y_\pi > 0$. Then
$(\rs(\pi)-1) \geq \frac{b+1}{2b+1}(r(\pi)-1)$.
\end{lemma}
\begin{proof}
For each part $\pi_i$ of $\pi$, let us identify all of the nodes of
$\pi_i$ into a single pseudonode $v_i$. We may assume by Theorem
\ref{thm:mst}
that each $T^*[\pi_i]$ is connected, hence this identification
process yields a tree $T'$. Let us say that $v_i$ is \emph{Steiner}
if and only if all nodes of $\pi_i$ are Steiner. Note that $T'$ has
$r(\pi)$ pseudonodes and $r(\pi)-\rs(\pi)$ of these pseudonodes are
Steiner. The full components of $T'$ are defined analogously to the
full components of a Steiner tree.

Consider any full component $K'$ of $T'$ and let $K'$ contain
exactly $s$ Steiner pseudonodes. It is straightforward to see that
$s \leq b$. Each Steiner pseudonode in $K'$ has degree at least 3 by
Assumptions A1 and A2, and at most $s-1$ edges of $K'$ join Steiner
vertices to other Steiner vertices. Hence $K'$ has at least
$3s-(s-1) = 2s+1$ edges, and so
$$|E(K')| \geq \frac{2s+1}{s}\cdot s
\geq \frac{2b+1}{b} \cdot s.
$$

Now summing over all full components $K'$, we obtain
$$|E(T')| \geq \frac{2b+1}{b} \cdot \#\{\textrm{Steiner pseudonodes of }T'\}.$$
But $|E(T')| = r(\pi)-1$ and $T'$ has $r(\pi)-\rs(\pi)$ Steiner
pseudonodes, therefore
$$r(\pi)-1 \geq \frac{2b+1}{b} ((r(\pi)-1)-(\rs(\pi)-1)) \quad \Rightarrow \quad \frac{2b+1}{b}(\rs(\pi)-1) \geq \frac{b+1}{b}(r(\pi)-1).$$
This proves what we wanted to show.
\end{proof}

It follows that the objective value of $y$ in \doobsp is
$$ \sum_{\pi \in \sppart[\ess]} (\rs(\pi) -1)y_{\pi}
   \geq \sum_{\pi \in \sppart[\ess]} \frac{b+1}{2b+1}(\rs(\pi) -1)y_{\pi}
   = \frac{b+1}{2b+1}c(T^*)$$
and since $T^*$ is an optimum integer solution of \poob, it follows
that the integrality gap of \poobsp is at most $\frac{b+1}{2b+1}$.
Then, finally, by applying Lemma \ref{lemma:gap++} to \newLPsp and
\poobsp we obtain Theorem \ref{thm:gap}.
\end{proof}

\section{Proof of Lemma \ref{lem:tpcost}}\label{sec:tpcost}

In this section we present a proof of Lemma \ref{lem:tpcost}. The
methodology follows that proposed by Gr\"opl et al.~\cite{GH+01b}.
In fact, many of the proofs below essentially correspond to those
presented in \cite{GH+01b} with two exceptions: we correct a small
error near the end, and we present a new proof of the ubiquitous
\emph{contraction lemma}.


We remind the reader of our standing assumption that $\ess \supseteq
\tbinom{R}{2}$. We first relate the cost of a minimum-cost spanning
tree of $\IS{\ess}$ for some set $\ess$ of full components to the
(potential) lower-bound $\smst(\ISC{\ess})$ on $\opt_r$ that it
provides. For ease of presentation in the analysis, we will assume
from now on that the costs of all edges in $E$ are pairwise
different. This assumption is easily seen to be w.l.o.g. (e.g., one
could define an order on the edges in $E$ and use it to break ties).
We omit the proof of the following easy fact.

\begin{fact} \label{factage}
If $T$ is a minimum-cost spanning tree of $\ess$ then $\loss(T) =
\loss(\ess)$.
\end{fact}

\begin{lemma}\label{lem:smst}
For any set $\ess \subseteq \fC_r$ of full components,
$$ \mst(\ISC{\ess}) = \smst(\ISC{\ess}) + \loss(\ess). $$
\end{lemma}
\begin{proof}
We use the notation from Section \ref{sec:mst}: $\tau^*$ is the
finishing time of Kruskal's algorithm, $G_\tau = (V, E_\tau)$ is the
forest maintained at time $\tau$, and $\pi_\tau$ is the partition
induced by the connected components of $G_\tau$. Let $(T, y)$ denote
the tree-dual pair returned by $\MST$.


From Theorem \ref{thm:mst} we know that there exists a feasible dual
solution $y$ to \eqref{d-sp} for graph $\IS{\ess}$ such that
$$
c(T) = \sum_{\pi \in \sppart[\ess]} (r(\pi)-1) y_{\pi} =
\int_0^{\tau^*} (r(\pi_{\tau})-1)d\tau.
$$
In the
following let $\calR_{\tau}$ be the set of those connected
components of $E_{\tau}$ that contain terminal vertices.

\begin{claim}\label{claim:rank}
  For all $0 \leq \tau \leq \tau^*$, each connected component of $E_{\tau}
  \cup \LOSS(T)$ contains exactly one connected component of $\calR_{\tau}$.
\end{claim}
\begin{proof}
Let $u$ and $v$ be terminals in distinct connected components of
$G_\tau$ and let $P_{uv}$ be the unique $u,v$-path in $T$.
Assume for the sake of contradiction that $P_{uv}$ is contained
in $E_{\tau} \cup \LOSS(T)$.

Let $\bar{e}$ be the unique edge of maximum cost on path $P_{uv}$.
Recall from Section \ref{sec:mst} that Kruskal's algorithm adds
edges to the partial spanning tree in order of non-decreasing cost.
Thus, edge $\bar{e}$ is added last among all edges on $P_{uv}$. As
$u$ and $v$ are in different connected components of $G_{\tau}$, it
therefore follows that $\bar{e}\not\in E_{\tau}$. The loss of $T$ is
a minimum-cost forest in $T$ that connects all Steiner vertices to
terminals. Thus, the unique edge of maximum cost on $P_{uv}$ cannot
be in $\LOSS(T)$.

It follows that $\bar{e} \not\in E_{\tau}\cup \LOSS(T)$ and this
contradicts our assumption that $P_{uv}
\subseteq E_{\tau} \cup \LOSS(T)$.
\end{proof}

For each time $0 \leq \tau\leq \tau^*$, define $\bar{\pi}_{\tau}$ as
the Steiner partition corresponding to the connected components of
$G_\tau \cup \LOSS(T)$.  From Theorem \ref{thm:mst} we know that
$$ \loss(T) = \sum_{e \in \LOSS(T)} c_e =
\sum_{e \in \LOSS(T)} \sum_{\pi: e\in E_{\pi}} y_{\pi} =
\int_0^{\tau^*} |E_{\pi_{\tau}} \cap \LOSS(T)| d\tau $$ where, as
before, $E_{\pi_{\tau}}$ is the set of edges in $E$ that have
endpoints in different parts of $\pi_{\tau}$.

The number of edges in $|E_{\pi_{\tau}} \cap \LOSS(T)|$ is exactly
the rank-difference between $\pi_{\tau}$ and $\bar{\pi}_{\tau}$ and
hence
$$ \loss(T) = \int_0^{\tau^*} (r(\pi_{\tau})-r(\bar{\pi}_{\tau})) d\tau. $$
Claim \ref{claim:rank} implies that
$r(\bar{\pi}_{\tau})=\rs(\pi_{\tau})$ for all $0 \leq \tau \leq
\tau^*$ and hence \begin{equation*} \smst(\ess) + \loss(T) =
\int_0^{\tau^*} (\rs(\pi)-1) d\tau +
    \int_0^{\tau^*} (r(\pi_{\tau})-\rs(\pi_{\tau})) d\tau
    = \int_0^{\tau^*} (r(\pi)-1) d\tau = c(T).
    \end{equation*}
Applying Fact \ref{factage} and the equality $c(T)=\mst(\ess)$, we
are done.
\end{proof}

We obtain the following immediate corollary:

\begin{corollary}\label{cor:fone-a}
In iteration $i$ of Algorithm \ref{alg}, adding full component $K
\in \fC_r$ to $\ess$ reduces the cost of $\mst(\ISC{\ess})$ if and
only if $f_i(K)<1$.
\end{corollary}
\begin{proof}
By applying Lemma \ref{lem:smst} we see that 
$$
\mst(\ISC{\ess^i}) - \mst(\ISC{\ess^i \cup \{K\}}) 
= \smst(\ISC{\ess^i}) + \loss(\ess^i)
     - \smst(\ISC{\ess^i \cup \{K\}})- \loss(\ess^i \cup\{K\}).$$
Whereas the left-hand side is positive iff adding $K$ to $\ess^i$
causes a reduction in $\mst$, the right-hand side is positive iff
$f_i(K)<1$, due to the definition of $f_i$.
\end{proof}


Using Lemma \ref{lem:feas} and Corollary \ref{cor:fone-a}, we obtain
the following.

\begin{corollary}\label{cor:fone}
For all $1 \leq i \leq p$, $f_i(K^i)<1$.
\end{corollary}

Fix an optimum $r$-Steiner tree $T^*$. \ignore{We would like to
assume that none of the full components added by Algorithm \ref{alg}
belong to $\fC_r(T^*)$.

To justify this assumption, for each $K \in \K$ let us create a
clone $K_{cl}$ of it, and add $K_{cl}$ to $\K$. We may assume that
$T^*$ uses \emph{only} cloned full components, and we claim that the
tree returned by \rz\ can be built using \emph{no} cloned full
components.

Indeed, the only problem would be if \rz\ tried to set $K^i$ to the
clone of $K^j$ for some $i$ and $j;$ but using Lemma \ref{lem:feas}
it is not hard to show that $f_j(K^i_{cl}) \geq 1$ for all $j>i$.}
The next two lemmas give bounds that are needed to analyze $\rz$'s
greedy strategy. Informally, the first says that $\smst$ is
non-increasing, while the second says that $\smst$ is submodular.

\begin{lemma}\label{lemma:smstdec}
If $\ess \subseteq \ess' \subseteq \K$, then $\smst(\ISC{\ess'})
\leq \smst(\ISC{\ess})$.
\end{lemma}
\begin{proof}
Using Lemma \ref{lem:smst} and Fact \ref{fact:modloss} we see
$$\smst(\ISC{\ess}) - \smst(\ISC{\ess'})  = \mst(\ISC{\ess})
+ \loss(\ess' \bs \ess) - \mst(\ISC{\ess'}).\label{smstdec:1}
$$ However, the right hand side of the above equation is
non-negative, as $ \MST(\ISC{\ess}) \cup\LOSS(\ess' \bs \ess)$ is a
spanning tree of $\ess'$. Lemma \ref{lemma:smstdec} then follows.
\end{proof}

\begin{lemma}[Contraction Lemma]\label{lem:subm}
Let $\R^0, \R^1, \R^2 \subset \K$ be disjoint collections of full
components with $\tbinom{R}{2} \subseteq \R^0$. Then
$$ \smst(\ISC{\R^0}) - \smst(\ISC{\R^0 \cup \R^2}) \geq
   \smst(\ISC{\R^0 \cup \R^1}) - \smst(\ISC{\R^0 \cup \R^1 \cup \R^2}). $$
\end{lemma}
\begin{proof}
The statement to be proved is equivalent to
\begin{equation} \mst(\ISC{\R^0}) - \mst(\ISC{\R^0 \cup \R^2}) \geq
   \mst(\ISC{\R^0 \cup \R^1}) - \mst(\ISC{\R^0 \cup \R^1 \cup \R^2}), \label{eq:reph}
   \end{equation}
due to Lemma \ref{lem:smst} and Fact \ref{fact:modloss}. For a
graph $H$, define the \emph{rank} $r(H)$ of $H$ as the number of
edges in a maximal forest of $H$:
$$ r(H) = |V(H)| - \mbox{\# connected components of } H. $$
For a graph $H$, let $H_{\leq x}$ denote the subgraph of $H$
consisting of those edges of weight at most $x$. By considering
Kruskal's algorithm, for any graph $H$ having nonnegative edge
costs, we see that
\begin{equation}\label{eq:intmst}
\mst(H) = \sum_{i=1}^{r(H)} \min \{x \mid r(H_{\leq x}) \geq i\}
=\int_{0}^{\infty} \bigl(r(H)-r(H_{\leq x})\bigr)\,dx.
\end{equation}
Note that the integral is proper since the integrand is 0 for $x$
larger than $\max\{c_e:e \in E(H)\}$.

Here is the crux: $r$ is the rank function for a (graphic) matroid
and is therefore submodular over the addition of disjoint edge sets.
Since the $\R^i_{<x}$ are pairwise disjoint, for every $x$, this
submodularity implies that
\begin{equation}
- r\left(\R^0_{\leq x}\right) + r\left(\R^0_{\leq x}\cup\R^2_{\leq
x}\right) \geq - r\left(\R^0_{\leq x}\cup\R^1_{\leq x}\right) +
r\left(\R^0_{\leq x}\cup\R^1_{\leq x}\cup\R^2_{\leq x}\right).
\label{subm:subm}
\end{equation}
Notice also that
\begin{equation}
r(\R^0) - r(\R^0\cup\R^2) =
r(\R^0\cup\R^1) - r(\R^0\cup\R^1\cup\R^2)\label{subm:eq}
\end{equation}
since both sides are equal to the number of Steiner vertices in $\R^2$, times $-1$.

Finally, we add Equation \eqref{subm:subm} to Equation
\eqref{subm:eq} and integrate along $x$. Since $(\R^0 \cup
\R^2)_{\leq x} = \R^0_{<x} \cup \R^2_{\leq x}$ etc.\ we get
\begin{align*}
& \int_0^\infty \left(r(\R^0)- r(\R^0_{\leq x})\right) \,dx -
\int_0^\infty \Bigl(r(\R^0 \cup \R^2)- r\left((\R^0 \cup \R^2)_{\leq x}\right)\Bigr)\,dx  \\
&\geq \int_0^\infty \Bigl(r(\R^0 \cup \R^1)- r\left((\R^0 \cup
\R^1)_{\leq x}\right)\Bigr)\,dx - \int_0^\infty \Bigl(r(\R^0 \cup
\R^1 \cup \R^2)- r\left((\R^0 \cup \R^1 \cup \R^2)_{\leq
x}\right)\Bigr)\,dx.
\end{align*}

But using Equation \eqref{eq:intmst}, this gives precisely Equation
\eqref{eq:reph}.
\end{proof}

We note that the proof of Lemma \ref{lemma:smstdec} easily
generalizes to other matroids. This is a departure from the existing
proofs in \cite{GH+01b} and \cite[Lemma 3.9]{BR94}, and Rizzi's more
specific result \cite[Lemma 2]{Ri03}, although a strong
\emph{exchange property} of matroids is used in the proof of
\cite{BR94}.

We are finally near the end of the analysis, where the Contraction
Lemma comes into play. We can now bound the value $f_i(K^i)$ for all
$0 \leq i \leq p-1$ in terms of the cost of $T^*$'s loss. In the
remainder of the section, let the full components of $T^*$ be
$K^{*,1}, \dotsc, K^{*,q}$, let $\loss^*$ denote $\loss(T^*)$, let
$\smst^i$ denote $\smst(\ISC{\ess^i})$ and let $\smst^*$ denote
$\smst(\ISC{T^*})$.

\begin{lemma}\label{lem:fbound}
For all $0 \leq i \leq p-1$, if $\smst^i - \smst^*
> 0$, then $f_i(K^i) \leq
\loss^*/(\smst^i - \smst^*)$.
\end{lemma}
\begin{proof}
By the choice of $K^i$ in Algorithm \ref{alg}, we have $f_i(K^i)
\leq \min_{j} f_i(K^{*,j})$. A standard fraction averaging
argument 
implies that
\begin{eqnarray}
f_i(K^i) & \leq & \frac{\sum_{j=1}^q \loss(K^{*,j})}
       {\sum_{j=1}^q \bigl(\smst(\ISC{\ess^i}) - \smst(\ISC{\ess^i \cup \{K^{*,j}\}})\bigr)} \notag \\
         & \leq & \frac{\loss^*}
       {\sum_{j=1}^q \bigl(\smst(\ISC{\ess^i \cup \{K^{*,1},\dotsc,K^{*,j-1}\}})
     - \smst(\ISC{\ess^i \cup
     \{K^{*,1},\dotsc,K^{*,j}\}})\bigr)} \label{coffee}
\end{eqnarray}
where the last inequality uses Fact \ref{fact:modloss} and Lemma
\ref{lem:subm}. (Additional care is needed when $T^*$ and $\ess^p$
overlap in some full components, but the above inequalities still
hold.) The denominator of the right-hand side of Equation
\eqref{coffee} is a telescoping sum. Canceling like terms, and using
Lemma \ref{lemma:smstdec} to replace $\smst(\ISC{\ess^i} \cup
\{K^{*,1},\dotsc,K^{*,q}\})$ with $\smst^*$, we are done.
\end{proof}

We can now bound the cost of $T^p$.

\begin{proof}[Proof of Lemma \ref{lem:tpcost}]
We first bound the loss $\loss(T^p)$ of tree $T^p$. Using Fact
\ref{fact:modloss}, \begin{equation}\label{tea} \loss(T^p) =
\sum_{i=0}^{p-1} \loss(K^i) =
                  \sum_{i=0}^{p-1} f_i(K^i) \cdot (\smst^i -
                  \smst^{i+1})\end{equation}
where the last equality uses the definition of $f_i$ from
\eqref{def:fi}. Using Corollary \ref{cor:fone} and Lemma
\ref{lem:fbound}, the right hand side of Equation \eqref{tea} is
bounded as follows:
\begin{equation}\label{tpcost:1}
\sum_{i=0}^{p-1} f_i(K^i) \cdot (\smst^i -
                  \smst^{i+1})\leq \sum_{i=0}^{p-1} \frac{\loss^*}{\max\{\loss^*, \smst^i - \smst^*\}}
   \cdot (\smst^i - \smst^{i+1}).
\end{equation}
The right hand side of Equation \eqref{tpcost:1} can in turn be
bounded from above by the following integral:
\begin{equation}\label{gummy} \sum_{i=0}^{p-1} \frac{\loss^* \cdot
(\smst^i - \smst^{i+1})}{\max\{\loss^* , \smst^i - \smst^*\}}
   \leq \int_{\smst^p}^{\smst^0} \frac{\loss^*}{\max\{\loss^*,x-\smst^*\}}
dx
   = \int_{\smst^p-\smst^*}^{\smst^0-\smst^*} \frac{\loss^*}{\max\{\loss^*,x\}}
   dx.\end{equation}
Notice that $\smst^0=\mst(G[R],c) \geq \opt_r = \loss^* + \smst^*$.
The termination condition in Algorithm \ref{alg} and Lemma
\ref{lem:genpart} imply that $\smst^p \leq \opt_r$. Hence the result
of evaluating the integral in the right-hand side of Equation
\eqref{gummy} is
\begin{equation} \label{crumpet} \loss^* - (\smst^p-\smst^*) + \loss^* \cdot
\int_{\loss^*}^{\smst^0-\smst^*} \frac{1}{x} dx =
   \opt_r - \smst^p + \loss^* \cdot
   \ln\left(\frac{\smst^0-\smst^*}{\loss^*}\right) \end{equation}
where the equality uses Lemma \ref{lem:smst}. Applying Lemma
\ref{lem:smst} two more times, and combining Equations
\eqref{tea}--\eqref{crumpet}, we obtain
\begin{eqnarray*}
 c(T^p) = \smst^p + \loss(T^p) & \leq & \opt_r + \loss^* \cdot
   \ln\left(\frac{\smst^0-\smst^*}{\loss^*}\right) \\
   & = & \opt_r + \loss^* \cdot
   \ln\left(1+ \frac{\smst^0-(\smst^*+\loss^*)}{\loss^*}\right) \\
   & = & \opt_r + \loss^* \cdot
   \ln\left( 1 + \frac{\smst^0-\opt_r}{\loss^*}\right)
\end{eqnarray*}
as wanted.
\end{proof}

{\bf Remark}. Gr\"{o}pl et al.\ essentially prove Lemma
\ref{lem:tpcost} in \cite[Lemma 4.3]{GH+01b} but a minor error lies
in their equation ``(18)." Namely, they assume ``$m_i - m^* > 0$"
which is $\smst^i-\smst^*
> 0$ in our notation.


\end{document}